\documentclass[11pt,a4paper]{llncs}
\usepackage[hmargin=2.75cm,top=2.9cm,bottom=3.3cm,nohead,footskip=40pt,marginparwidth=2.8cm,marginparsep=.2cm]{geometry}
\usepackage{mathptmx} 
\usepackage[english]{babel}
\usepackage[applemac]{inputenc} 
\usepackage{graphicx}
\usepackage{amsmath,amssymb,amsfonts,amsbsy}
\usepackage{array}
\usepackage{txfonts}
\usepackage{tabularx}
\usepackage{float}
\usepackage{empheq}
\usepackage{color}
\usepackage{cancel}
\usepackage{pstricks, pstricks-add, pst-node, pst-tree, pst-plot, pst-3d} 

\newcommand{\R}{\mathbb{R}} \newcommand{\N}{\mathbb{N}}

\newcommand{\SNE}{X^{\text{SNE}}}
\title{\normalfont\LARGE\bfseries  Strong Nash Equilibria in Games
  with the \\ Lexicographical Improvement Property}

\author{Tobias Harks \and Max Klimm \and Rolf~H.~M\"ohring
\\ Preprint 15/07/2009}
\institute{
  Institut f\"ur Mathematik, Technische Universit\"at Berlin \\
  \email{\{harks,klimm,moehring\}@math.tu-berlin.de} }

\begin{document}

\maketitle
\begin{abstract}
  We introduce a class of finite strategic games with the property
  that every deviation of a coalition of players that is profitable to
  each of its members strictly decreases the lexicographical order of
  a certain function defined on the set of strategy profiles. We call
  this property the \emph{Lexicographical Improvement Property (LIP)}
  and show that it implies the existence of a generalized strong
  ordinal potential function. We use this characterization to derive
  existence, efficiency and fairness properties of strong Nash
  equilibria.
  
  We then study a class of games that generalizes congestion games
  with bottleneck objectives that we call \emph{bottleneck congestion
    games}.  We show that these games possess the LIP and thus the
  above mentioned properties.  For bottleneck congestion games in
  networks, we identify cases in which the potential function
  associated with the LIP leads to polynomial time algorithms
  computing a strong Nash equilibrium.

  Finally, we investigate the LIP for infinite games. We show that the
  LIP does not imply the existence of a generalized strong ordinal
  potential, thus, the existence of SNE does not follow. Assuming that
  the function associated with the LIP is continuous, however, we
  prove existence of SNE.  As a consequence, we prove that bottleneck
  congestion games with infinite strategy spaces and continuous cost
  functions possess a strong Nash equilibrium.
  \end{abstract}

\section{Introduction}
The theory of non-cooperative games is used to study situations that
involve rational and selfish agents who are motivated by optimizing
their own utilities rather than reaching some social optimum.  In such
a situation, a state is called a pure Nash equilibrium (PNE) if it is
stable in the sense that no player has an incentive to unilaterally
deviate from its strategy. While the PNE concept excludes the
possibility that a single player can unilaterally improve her utility,
it does not necessarily imply that a PNE is stable against coordinated
deviations of a group of players if their joint deviation is
profitable for each of its members.  So when coordinated actions are
possible, the Nash equilibrium concept is not sufficient to analyze
stable states of a game.

To cope with this issue of coordination, we adopt the solution concept
of a \emph{strong equilibrium} (SNE for short) proposed by
Aumann~\cite{Aumann59}. In a strong equilibrium, no coalition (of any
size) can deviate and strictly improve the utility of each of it
members (while possibly lowering the utility of players outside the
coalition).  Clearly, every SNE is a PNE, but the converse does not
always hold. Thus, even though SNE may rarely exist, they form a very
robust and appealing stability concept.

One of the most successful approaches in establishing existence of PNE
(as well as SNE) is the potential function approach initiated by
Monderer and Shapley~\cite{Monderer:1996} and later generalized by
Holzman and Law-Yone~\cite{Holzman97} to strong equilibria: one
defines a real-valued function $P$ on the set of strategy profiles of
the game and shows that every improving move of a coalition (which is
profitable to each of its members) strictly reduces (increases) the
value of $P$.  Since the set of strategy profiles of such a (finite)
game is finite, every sequence of improving moves reaches a SNE. In
particular, the global minimum (maximum) of $P$ is a SNE. The main
difficulty is, however, that for most games it is hard to prove or
disprove the existence of such a potential function.

Given that the existence of a real-valued potential is hard to detect,
we derive in this paper an equivalent property
(Theorem~\ref{theorem:equivalent}) that we call the
\emph{lexicographical improvement property (LIP)}. We consider
strategic games $G = (N,X,\pi)$, where $N$ is the set of players, $X$
the strategy space, and players experience private non-negative costs
$\pi_i(x),i\in N$, for a strategy profile $x$.  We say that $G$ has
the LIP if there exists a function $\phi:X\rightarrow\R_+^q, q\in \N,$
such that every improving move (profitable deviation of an arbitrary
coalition) from $x\in X$ strictly reduces the sorted lexicographical
order of $\phi(x)$ (see Definition~\ref{definition:lip}). We say that
$G$ has the $\pi$-LIP if $G$ satisfies the LIP with
$\phi_i(x)=\pi_i(x), i\in N$, that is, every improving move strictly
reduces the lexicographical order of the players' private costs.
Clearly, requiring $q=1$ in the definition of the LIP reduces to the
case of a generalized strong ordinal potential.

The main contribution of this paper is twofold. We first study
desirable properties of arbitrary finite and infinite games having the
LIP and $\pi$-LIP, respectively. These properties concern the
existence of SNE, efficiency and fairness of SNE, and computability of
SNE. Secondly, we identify an important class of games that we term
bottleneck congestion games for which we can actually prove the
$\pi$-LIP and, hence, prove that these games possess SNE with the
above desirable properties. In the following, we will give an informal
definition of bottleneck congestion games.

Let us first recall the definition of a standard congestion game. In a
congestion game, there is a set of elements (facilities), and the pure
strategies of players are subsets of this set.  Each facility has a
cost that is a function of its \emph{load}, which is usually a
function of the number (or total weight) of players that select
strategies containing the respective facility.  The private cost of a
player's strategy in a standard congestion game is given by the
\emph{sum} over the costs of the facilities in the strategy.  In a
bottleneck congestion game, the private cost function of a player is
equal to the cost of the \emph{most expensive} facility that she uses
($L_{\infty}$-norm of the vector of players' costs of the strategy).

Before we outline our results, we briefly explain the importance of
bottleneck objectives in congestion games with respect to real-world
applications.  Referring to previous work by Keshav~\cite{Keshav97},
Cole et al.~\cite{ColeDR06} pointed out that the performance of a
communication network is closely related to the performance of its
bottleneck (most congested) link.  This behavior is also stressed by
Banner and Orda~\cite{BannerO07}, who studied Nash equilibria in
routing games with bottleneck objectives.  Similar observations are
reported by Qiu et al.~\cite{Qiu06}, who investigated the
applicability of theoretical results of selfish routing models to
realistic models of the Internet.

\subsection{Our Results}
We characterize games having the LIP by means of the existence of a
generalized strong ordinal potential function. The proof of this
characterization is constructive, that is, given a game $G$ having the
LIP for a function $\phi$, we explicitely construct a generalized
strong ordinal potential $P$.  We further investigate games having the
$\pi$-LIP with respect to efficiency and fairness of SNE.  Our
characterization implies that there are SNE satisfying various
efficiency and fairness properties, e.g., bounds on the prices of
stability and anarchy, Pareto optimality, and min-max fairness.

We establish that bottleneck congestion games possess the $\pi$-LIP
and, thus, possess SNE with the above mentioned properties.  Moreover,
our characterization of games having the LIP implies that bottleneck
congestion games possess the strong finite improvement property. For
bottleneck congestion games in networks, we identify cases in which
SNE can be computed in polynomial time.

It is worth noting that for singleton congestion games, Even-Dar et
al.~\cite{Even07} and Fabrikant et al.~\cite{Fabrikant04}, have
already proved existence of PNE by arguing that the vector of facility
costs decreases lexicographically for every improving move. Andelman
et al.~\cite{Andelman09} used the same argument to establish even
existence of SNE in this case. Our work generalizes these results to
arbitrary strategy spaces and more general cost functions.  In
contrast to most congestion games considered so far, we require only
that the cost functions on the facilities satisfy three properties:
"non-negativity", "independence of irrelevant choices", and
"monotonicity". Roughly speaking, the second and third condition
assume that the cost of a facility solely depends on the set of
players using the respective facility and that this cost decreases if
some players leave this facility.  Thus, this framework extends
classical \emph{load-based} models in which the cost of a facility
depends on the number or total weight of players using the respective
facility.  Our assumptions are weaker than in the load-based models
and even allow that the cost of a facility may depend on the
\emph{set} of players using this facility.

We then study the LIP in \emph{infinite} games, that is, games with
infinite strategy spaces that can be described by compact subsets of
$\R^p,\,p\in \N$. We first show that our characterization of finite
games with the LIP does not hold anymore (essentially resembling
Debreu's impossibility result~\cite{Debreu:1954}). We prove, however,
that continuity of $\phi$ in the definition of LIP is sufficient for
the existence of SNE.  Our existence proof is constructive, that is,
we outline an algorithm whose output is a SNE.

We consequently introduce \emph{infinite bottleneck congestion games}.
An infinite bottleneck congestion game arises from a bottleneck
congestion game $G$ by allowing players to fractionally distribute a
certain \emph{demand} over the pure strategies of $G$.  We prove that
these games have the $\pi$-LIP provided that the cost functions on the
facilities are non-negative and non-decreasing.  It turns out,
however, that the function $\pi$ may be discontinuous on the strategy
space (even if the cost functions on the facilities are continuous).
Thus, the existence of SNE does not immediately follow.  We solve this
difficulty by generalizing the LIP.  As a consequence, we obtain for
the first time the existence of SNE for infinite bottleneck congestion
games with non-decreasing and continuous cost functions.  For
\emph{bounded} cost functions on the facilities (that may be
discontinuous), we show that $\alpha$-approximate SNE exist for every
$\alpha>0$. Finally, we show that $\alpha$-approximate SNE can be
computed in polynomial time for bottleneck congestion games in
networks.

In the final section, we show that our methods presented in this paper
also apply to a more general framework.

\subsection{Related Work}

The SNE concept was introduced by Aumann~\cite{Aumann59} and refined
by Bernheim et al.~\cite{Bernheim87} to Coalition-Proof Nash
Equilibrium (CPNE), which is a state that is stable against those
deviations, which are themselves resilient to further deviations by
subsets of the original coalition.  This implies that every SNE is
also a CPNE, but the converse need not hold.

Congestion games were introduced by Rosenthal~\cite{Rosenthal:1973}
and further studied by Monderer and Shapley~\cite{Monderer:1996}.
Holzman and Law-Yone~\cite{Holzman97} studied the existence of SNE in
congestion games with monotone increasing cost functions. They showed
that SNE need not exist in such games and gave a structural
characterization of the strategy space for symmetric (and
quasi-symmetric) congestion games that admit SNE. Based on the
previous work of Monderer and Shapley~\cite{Monderer:1996}, they also
introduced the concept of a \emph{strong potential function}: a
function on the set of strategy profiles that decreases for every
profitable deviation of a coalition.  Rozenfeld and
Tennenholtz~\cite{RozenfeldT06} further explored the existence of
(correlated) SNE in congestion games with non-increasing cost
functions.

A further generalization of congestion games has been proposed by
Milchtaich~\cite{Milchtaich:1996}, where he allows for player-specific
cost functions on the facilities (see also Mavronicolas et
al.~\cite{Mavronicolas:2007}, Gairing et al.~\cite{GairingMT06} and
Ackermann et al.~\cite{Ackermann09} for subsequent work on weighted
congestion games with player-specific cost functions). Under
restrictions on the strategy space (singleton strategies), Milchtaich
proves existence of pure Nash equilibria. As shown by Voorneveld et
al.~\cite{Voorneveld:1999}, the model of Konishi et
al.~\cite{Konishi97a} is equivalent to that of Milchtaich, which is
worth noting as Konishi et al. established the existence of SNE in
such games.

Several authors studied the existence and efficiency (price of anarchy
and stability) of PNE and SNE in various specific classes of
congestion games.  For example, Even-Dar et al.~\cite{Even07} showed
that job scheduling games (on unrelated machines) admit a PNE by
arguing that the load-lexicographically minimal schedule is a PNE.
Fabrikant et al.~\cite{Fabrikant04} considered a scheduling model in
which the processing time of a machine may depend on the set of jobs
scheduled on the respective machine. For this model, they proved
existence of PNE analogous to the proof of Even-Dar et al.  Andelman
et al.~\cite{Andelman09} considered scheduling games on unrelated
machines and proved that the load-lexicographically minimal schedule
is even a SNE.  They further studied differences between PNE and SNE
and derived bounds on the (strong) price of anarchy and stability,
respectively. Chien and Sinclair~\cite{Chien09} recently studied the
strong price of anarchy of SNE in general congestion games.

Bottleneck congestion games with network structure have been
considered by Banner and Orda~\cite{BannerO07}. They studied existence
of PNE in the unsplittable flow and in the splittable flow setting,
respectively. They observed that standard techniques (such as
Kakutani's fixed-point theorem) for proving existence of PNE do not
apply to bottleneck routing games as the private cost functions may be
discontinuous. They proved existence of PNE by showing that bottleneck
games are better reply secure, quasi-convex, and compact.  Under these
conditions, they recall Reny's existence theorem~\cite{Reny99} for
better reply secure games with possibly discontinuous private cost
functions.  Banner and Orda, however, do not study SNE. We remark that
our proof of the existence of SNE is direct and constructive.

Bottleneck routing with \emph{non-atomic} players and elastic demands
has been studied by Cole et al.~\cite{ColeDR06}.  Among other results
they derived bounds on the price of anarchy in this setting. For
subsequent work on the price of anarchy in bottleneck routing games
with atomic and non-atomic players, we refer to the paper by Mazalo et
al.~\cite{MazalovMST06}.

\section{Preliminaries}

We consider strategic games $G = (N,X,\pi)$, where $N=\{1,\dots,n\}$
is the non-empty and finite set of players, $X = \varprod_{i \in N}
X_i$ is the non-empty strategy space, and $\pi : X \to
\mathbb{R}^{n}_{+}$ is the combined \emph{private cost} function
that assigns a private cost vector $\pi(x)$ to each strategy profile $x \in
X$.  These games are cost minimization games and we assume
additionally that the private cost functions are non-negative. A
strategic game is called \emph{finite} if $X$ is finite.

We use standard game theory notation; for a coalition $S\subseteq N$
we denote by $-S$ its complement and by $X_S = \varprod_{i \in S} X_i$
we denote the set of strategy profiles of players in $S$.

\begin{definition}[Strong Nash equilibrium (SNE)]
  A strategy profile $x$ is a strong Nash equilibrium if there is no
  coalition $\emptyset \neq S\subseteq N$ that has an alternative
  strategy profile $y_S \in X_S$ such that
  $\pi_i(y_S,x_{-S})-\pi_i(x)<0$ for all $i\in S$.
\end{definition}

A pair $\bigl(x,(y_S,x_{-S})\bigr) \in X\times X$ is called an
\emph{improving move} (or \emph{profitable deviation}) of coalition
$S$ if $\pi_i(x_S,x_{-S}) - \pi_i(y_S,x_{-S}) > 0 \text{ for all }i\in
S$. We denote by $I(S)$ the set of improving moves of coalition $S
\subseteq N$ in a strategic game $G = (N,X,\pi)$ and we set $I :=
\bigcup_{S \subseteq N} I(S)$.  We call a sequence of strategy
profiles $\gamma=(x^0,x^1,\dots)$ an \emph{improvement path} if every
tuple $(x^{k},x^{k+1})\in I$ for all $k = 0,1,2,\dots$.  One can
interpret an improvement path as a path in the so called
\emph{improvement graph $\mathcal{G}(G)$} derived from $G$, where
every strategy profile $x\in X$ corresponds to a node in
$\mathcal{G}(G)$ and two nodes $x, x'$ are connected by a directed
edge $(x,x')$ if and only if $(x,x')\in I$. We are interested in
conditions that assure that every improvement path is finite. A
necessary and sufficient condition is the existence of a generalized
\emph{strong} ordinal potential function, which we define below (see
also the potential function approach initiated by Monderer and
Shapley~\cite{Monderer:1996}, which has been generalized to strong
potentials by Holzman and Law-Yone~\cite{Holzman97}).
\begin{definition}[Generalized strong ordinal potential game]
\label{definition:gen_strong_ordinal_potential}
A strategic game $G = \left(N, X, \pi\right)$ is called a
\emph{generalized strong ordinal potential game} if there is a
function $P : X \to \mathbb{R}$ such that $P(x) - P(y) > 0$ for all
$(x,y) \in I$. $P$ is called a generalized strong ordinal potential of
the game $G$.
\end{definition}
In recent years, much attention has been devoted to games admitting
the finite improvement property (FIP), that is, each path of
single-handed (one player) deviations is finite. Equivalently, we say
that $G$ has the strong finite improvement property (SFIP) if every
improvement path is finite. Clearly, the SFIP implies the FIP, but the
converse need not be true.

It is known that both the SFIP and the existence of a generalized
strong ordinal potential are hard to prove or disprove for a
particular game. We define a class of games that we call \emph{games
  with the Lexicographical Improvement Property (LIP)} and show that
such games possess a generalized strong ordinal potential. For this
purpose, we will first define the sorted lexicographical order.
\begin{definition}[Sorted lexicographical order]
\label{definition:lexicographic_sorting}
Let $a,b\in\R_+^q$ and denote by $\tilde{a},\tilde{b}\in\R_+^q$ be the
sorted vectors derived from $a,b$ by permuting the entries in
non-increasing order, that is, $\tilde{a}_1\geq\cdots\geq \tilde{a}_q$
and $\tilde{b}_1\geq\cdots\geq \tilde{b}_q$. Then, $a$ is
\emph{strictly sorted lexicographically smaller} than $b$ (written $a\prec
b$) if there exists an index $m$ such that $\tilde{a}_i=\tilde{b}_i$
for all $i< m$, and $\tilde{a}_m<\tilde{b}_m$. The vector $a$ is
\emph{sorted lexicographically smaller} than $b$ (written $a\preceq b$) if
either $a\prec b$ or $\tilde{a}=\tilde{b}$.
\end{definition}

Roughly speaking, the lexicographical improvement property of a
strategic game requires that a vector-valued function $\phi : X \to
\mathbb{R}_+^q$ is strictly decreasing with respect to the sorted
lexicographical order on $\mathbb{R}_+^q$ for every improvement step.

\begin{definition}[Lexicographical improvement property, $\pi$-LIP]
  \label{definition:lip}
A finite strategic game $G= \left(N, X, \pi\right)$ possesses the
\emph{lexicographical improvement property} (LIP) if there exist $q\in
\N$ and a function $\phi:X\rightarrow \R^q_+$ such that $\phi(x) \succ
\phi(y)$ for all $(x,y)\in I$. $G$ has the $\pi$-LIP if $G$ has the
LIP for $\phi=\pi$.
\end{definition}

Clearly, the function $\phi$ is a generalized strong ordinal potential
if $q=1$. The next theorem (proof is given in the appendix) states
that requiring the LIP is equivalent to requiring the existence of a
generalized strong ordinal potential, regardless of $q$.

\begin{theorem}
\label{theorem:equivalent}
Let $G= \left(N, X, \pi\right)$ be a finite strategic game.  Then, the
following statements are equivalent.
\begin{enumerate}
\item G has the SFIP
\item $\mathcal{G}(G)$ is acyclic (contains no directed cycles)
\item G has a generalized strong ordinal potential
\item $G$ has the LIP
\item There exists $\phi:X\rightarrow \R^q_+$ and $M \in \N$ such that
  $P(x)=\sum_{i=1}^q\phi_i(x)^{M} $ is a generalized strong ordinal
  potential function for $G$.
\end{enumerate}
\end{theorem}

 \begin{corollary}\label{cor:sne-ex}
  Every finite strategic game $G = (N,X,\pi)$ with the LIP possesses a
  strong Nash equilibrium.
\end{corollary}

Next, we provide an explicit formula to obtain a generalized strong
ordinal potential function for a strategic game satisfying the LIP.
The proof is given in the appendix.

\begin{corollary}\label{cor:ord-pot}
  Let $G = (N,X,\pi)$ be a strategic game that satisfies the LIP for a
  function $\phi : X \to \mathbb{R}_+^q$. We set
\begin{align*}
  \phi_{\max} := \max_{x \in X, 1 \leq i \leq q} \phi_i(x),\;\;
  \epsilon_{\min} := \mathop{\min_{(x,y) \in I}}_{1\leq i \leq q :
    \phi_i(x) \neq \phi_i(y)} \phi_i(x) - \phi_i(y),
\end{align*}
and choose $M> \log(q)\,\phi_{\max}\, /\, \epsilon_{\min}$. Then,
$P_M(x)=\sum_{i=1}^q\phi_i(x)^{M}$ is a generalized strong ordinal
potential for $G$.
\end{corollary}

\begin{corollary}\label{cor:number-of-moves}
  Let $G = (N,X,\pi)$ be a finite game satisfying the LIP for a
  function $\phi : X \to \mathbb{R}_+^q$.  Then, the number of arcs
  along any improvement path in $\mathcal{G}(G)$ is bounded from above
  by $ \lceil q\,\phi_{\max}^M/\epsilon_{\min}\rceil,$ where
  $M>\log(q) \frac{\phi_{\max}}{\epsilon_{\min}}$.
\end{corollary}

\section{Properties of SNE in Games with the $\pi$-LIP}
As the existence of SNE in games with the LIP is guaranteed, it is
natural to ask which properties these SNE may satisfy. In recent
years, several notions of efficiency have been discussed in the
literature, see Koutsoupias and Papadimitriou~\cite{KoP99}. We here
cover the price of stability, the price of anarchy, Pareto optimality
and min-max-fairness.
\subsection{Price of Stability and Price of Anarchy}
We study the efficiency of SNE with respect to the optimum of a
predefined social cost function. In this context, two notions have
evolved, the \emph{strong price of anarchy} measures the ratio of the
cost of the worst SNE and that of the social optimum. The \emph{strong
  price of stability} measures the ratio of the cost of the best SNE
and that of the social optimum. Given a game $G = (N,X,\pi)$ and a
social cost function $C : X \to \mathbb{R}_{+},$ whose minimum is
attained in a strategy profile $y \in X$, let $\SNE \subseteq X$\
denote the set of strong Nash equilibria. Then, the strong price of
anarchy for $G$ with respect to $C$ is defined as $\sup_{x\in \SNE}
C(x) /C(y)$ and the strong price of stability for $G$ with respect to
$C$ is defined as $\inf_{x\in\SNE}C(x) /C(y)$.

We will consider the following natural social cost functions: the
sum-objective or $L_1$-norm defined as $L_1(x)=\sum_{i\in N}\pi_i(x)$,
the $L_p$-objective or $L_p$-norm, $p\in\N$, defined as
$L_p(x)=\big(\sum_{i\in N}\pi_i(x)^p\big)^{1/p}$, and the min-max
objective or $L_\infty$-norm defined as $L_\infty(x)=\max_{i\in N}
\{\pi_i(x)\}$.

\begin{theorem}\label{thm:stab}
  Let $G$ be a strategic game with the $\pi$-LIP. Then, the strong
  price of stability w.r.t. $L_\infty$ is $1$, and for any $p \in
  \mathbb{R}$, the strong price of stability w.r.t. $L_p$ is strictly
  smaller than $n$.
\end{theorem}
The proof of this theorem as well as a matching lower bound for $p=1$
and a lower bound on the price of anarchy are given in the appendix.

\subsection{Pareto Optimality}
Pareto optimality, which we define below, is one of the fundamental
concepts studied in economics, see Osborne and
Rubinstein~\cite{osborne-rubinstein94}. For a strategic game
$G=(N,X,\pi)$, a strategy profile $x$ is called \emph{weakly Pareto
  efficient} if there is no $y\in X$ such that $ \pi_i(y)< \pi_i(x)
\text{ for all }i\in N.$ A strategy profile $x$ is \emph{strictly
  Pareto efficient} if there is no $y\in X$ such that $ \pi_i(y)\leq
\pi_i(x) \text{ for all }i\in N,$ where at least one inequality is
strict.

So strictly Pareto efficient strategy profiles are those strategy
profiles for which every improvement of a coalition of players is to
the expense of at least one player outside the coalition. Pareto
optimality has also been studied in the context of congestion games,
see Chien and Sinclair~\cite{Chien09} and Holzman and
Law-Yone~\cite{Holzman97}.  Clearly, every SNE is weakly Pareto
optimal. We will show strict Pareto optimality of SNE in games having
the $\pi$-LIP. For the proof of the next result, we refer to
Section~\ref{subsec:min-max} in which an even stronger result is
proved.
\begin{corollary}
  Let $G$ be a finite strategic game having the $\pi$-LIP.  Then,
  there exists a SNE that is strictly Pareto optimal.
\end{corollary}
\subsection{Min-Max-Fairness}\label{subsec:min-max}
We next define the notion of min-max fairness, which is a central
topic in resource allocation in communication networks, see
Srikant~\cite{srikant03} for an overview and pointers to the large
body of research in this area. While strict Pareto efficiency requires
that there is no improvement to the expense of anyone, the notion of
min-max-fairness is stricter. Here, it is only required that there is
no improvement at the cost of someone who receives already higher
costs (while an improvement that increases the cost of a player with
smaller original cost is allowed). It is easy to see that every
min-max-fair strategy profile constitutes a strict Pareto optimum, but
the converse need not hold. A strategy profile $x$ is called
\emph{min-max fair} if for any other strategy profile $y$ with
$\pi_i(y)<\pi_i(x)$ for some $i\in N,$ there exists $j\in N$ such that
$\pi_j(x)\geq \pi_i(x)$ and $\pi_j(y)> \pi_j(x)$.
\begin{corollary}\label{cor:min-max}
  Let $G$ be a finite strategic game having the $\pi$-LIP.  Then,
  there exists a SNE that is min-max fair.
\end{corollary}
The corollary is proved in the appendix.

\section{Bottleneck Congestion Games}\label{sec:allocation}
We now present a rich class of games satisfying the $\pi$-LIP.  We
call these games \emph{bottleneck congestion games}.  They are natural
generalizations of variants of congestion games.  In contrast to
standard congestion games, we focus on \emph{makespan-objectives},
that is, the cost of a player when using a set of facilities only
depends on the highest cost of these facilities. For the sake of a
clean mathematical definition, we introduce the general notion of a
congestion model.
\begin{definition}[Congestion model]
  A tuple $\mathcal{M} = (N, F, X, (c_f)_{f \in F})$ is called a
  \emph{congestion model} if $N = \{1,\dots,n\}$ is a non-empty,
  finite set of players, $F = \{1,\dots,m\}$ is a non-empty, finite
  set of facilities, and $X = \varprod_{i \in N} X_i$ is the set of
  strategies. For each player $i \in N$, her collection of pure
  strategies $X_i$ is a non-empty, finite set of subsets of $F$. Given
  a strategy profile $x$, we define $\mathcal{N}_f(x)=\{i\in N \,:\,
  f\in x_i\} \text{ for all }f\in F.$ Every facility $f\in F$ has a
  cost function $c_f :\varprod_{i \in N} X_i \rightarrow\R_+ $
  satisfying
\begin{description}
\item[Non-negativity:] $c_f(x)\geq 0 \text{ for all }x\in X$
\item[Independence of Irrelevant Choices:] $c_f(x)=c_f(y) \text{ for
    all }x,y\in X \text{ with } \mathcal{N}_f(x)=\mathcal{N}_f(y)$
\item[Monotonicity:] $c_f(x)\leq c_f(y) \text{ for all }x,y\in X
  \text{ with } \mathcal{N}_f(x)\subseteq\mathcal{N}_f(y)$.
 \end{description}

  \end{definition}
  Bottleneck congestion games generalize congestion games in the
  definition of the cost functions on the facilities. For bottleneck
  congestion games, the only requirements are that the cost $c_f(x)$
  on facility $f$ for strategy profile $x$ only depends on the set of
  players using $f$ in their strategy profile and that costs are
  increasing with larger sets.

\begin{definition}[Bottleneck congestion game]
\label{def:allocation}
Let $\mathcal{M} = (N, F, X, (c_f)_{f \in F})$ be a congestion model.
The corresponding \emph{bottleneck congestion game} is the strategic
game $G(\mathcal{M})~=~(N, X, \pi)$ in which $\pi$ is defined as
$\pi=\varprod_{ i \in N} \pi_i$ and $\pi_i(x) = \max_{f\in x_i}
c_f\big( x \big).$
\end{definition}

A bottleneck congestion game with $|x_i|=1$ for all $x_i \in X_i$ and
$i \in N$ will be called \emph{singleton} bottleneck congestion game.
We remark that our condition "Independence of Irrelevant Choices" is
weaker than the one frequently used in the literature. In Konishi et
al.~\cite{Breton96,Konishi97a,Konishi97b}, the definition of
"Independence of Irrelevant Choices" requires that the strategy spaces
are symmetric and, given a strategy profile $x=(x_1,\dots,x_n)$, the
utility of a player $i$ depends only on her own choice $x_i$ and the
cardinality of the set of other players who also choose $x_i$.  Our
definition does neither require symmetry of strategies, nor that the
utility of player $i$ only depends on the set-cardinality of other
players who also choose $x_i$. For the relationship between games
considered by Konishi et al. and congestion games, see the discussion
in Voorneveld et al.~\cite{Voorneveld:1999}.

We are now ready to state our main result concerning bottleneck
congestion games, providing a large class of games that satisfies the $\pi$-LIP.
\begin{theorem}\label{thm:all}
  Let $G(\mathcal{M})$ be a bottleneck congestion game with allocation
  model $\mathcal{M}$. Then, $G$ fulfills the
  LIP for the functions $\phi : X \to \mathbb{R}_+^n$ and $\psi : X
  \to \mathbb{R}_+^{m\,n}$ defined as
\begin{align*}
\phi_i(x) = \pi_i(x) \quad\text{ for all } i \in N, &&
\psi_{i,f}(x) = 
\begin{cases}
c_f(x) &\text{if $f\in x_i$}\\
0 &\text{else}
\end{cases} \quad\text{ for all } i\in N, f \in F.
\end{align*}
\end{theorem}
The proof is given in the appendix.

  Note that the function $\upsilon : X \to \mathbb{R}^m$ with $x
  \mapsto \bigl(c_f(x)\bigr)_{f \in F}$ does not fulfill $\upsilon(x)
  \succ \upsilon(y)$ for all $(x,y) \in I$. However, this property is
  satisfied if facility cost functions are strictly monotonic.
  
  As a corollary of Theorem \ref{thm:all} we obtain that each
  bottleneck congestion game has the $\pi$-LIP and hence possesses the SFIP.
  In addition, the results on price of stability, Pareto optimality
  and min-max-fairness apply. 
We now proceed by giving
some examples of bottleneck congestion games.

\subsection{Scheduling Games}

A special case of bottleneck congestion games are scheduling games on
identical, restricted, related and unrelated machines. These games are
allocation games by restricting the strategy space for every player to
singletons and defining the appropriate cost functions on the
facilities.

\begin{corollary}
  Scheduling games on identical, restricted, related and unrelated
  machines are bottleneck congestion games.
\end{corollary}
The existence of SNE (as well as efficiency properties of SNE) has
been established before by Andelman et al.~\cite{Andelman09} by
arguing that the lexicographically minimal schedule is a strong Nash
equilibrium.

\subsection{Resource Allocation in Wireless Networks}
Interference games are motivated by resource allocation problems in
wireless networks. Consider a set of $n$ terminals that want to
connect to one out of $m$ available base stations.  Terminals assigned
to the same base station impose interferences among each other as they
use the same frequency band. We model the interference relations using
an undirected interference graph $\mathcal{D}=(V,E)$, where
$V=\{1,\dots,n\}$ is the set of vertices/terminals and an edge
$e=(v,w)$ between terminals $v,w$ has a non-negative weight $w_{e}\geq
0$ representing the level of pair-wise interference.  We assume that
the service quality of a base station $j$ is proportional to the total
interference $w(j)$, which is defined as $w(j)=\sum_{(v,w)\in E:
  x_v=x_w=j}w_{(v,w)}.$

We now obtain an interference game as follows. The nodes of the graph
are the players, the set of strategies is given by $X_i=\{1,\dots,m\},
i=1,\dots,n$, that is, the set of base stations, and the private cost
function for every player is defined as $ \pi_i(x)=w(x_i), \;
i=1,\dots,n.$ Thus, interference games fit into the framework of
singleton bottleneck congestion games with $m$ facilities.
\begin{corollary}
Interference games are bottleneck congestion games.
\end{corollary}
Note that in interference games, we crucially exploit the property
that cost functions on facilities depend on the set of players using
the facility, that is, their \emph{identity} determines the resulting
cost. Most previous game-theoretic works addressing wireless networks
only considered Nash equilibria, see for instance Liu et
al.~\cite{Wu08} and Etkin et al.~\cite{Etkin07}.

\subsection{Bottleneck Routing in Networks}
A special case of bottleneck congestion games are bottleneck routing
games. Here, the set of facilities are the edges of a directed graph
$\mathcal{D}=(V,A)$. Every edge $a\in A$ has a load dependent cost
function $c_a$.  Every player is associated with a pair of vertices
$(s_i,t_i)$ (commodity) and a fixed demand $d_i>0$ that she wishes to
send along the chosen path in $\mathcal{D}$ connecting $s_i$ to $t_i$.
The private cost for every player is the maximum arc cost along the
path.  Bottleneck routing games have been studied by Banner and
Orda~\cite{BannerO07}.  They, however, did not study existence of
strong Nash equilibria.  We state the following result.
\begin{corollary}
  Bottleneck routing games are bottleneck congestion games.
\end{corollary}
This result establishes that these games possess the SFIP.  To the
best of our knowledge, our result establishes for the first time that
bottleneck routing games possess the FIP.  Banner and
Orda~\cite{BannerO07} only proved that every improvement path of
best-response dynamics is finite.

By using a reduction from \emph{k-Directed Disjoint Paths}, Banner and
Orda~\cite{BannerO07} proved that, given a value $B$, it is NP-hard to
decide if a bottleneck routing game with $k$ commodities and identical
cost functions possesses a PNE $x$ with $L_{\infty}(x)\leq B$.  We can
slightly strengthen this result by reducing from \emph{2-Directed
  Disjoint Paths} showing hardness already for $2$ commodities. While
every SNE is also a PNE, the above hardness result also carries over
to SNE.

We will show that for single-commodity instances with unit demands and
non-decreasing identical cost functions, there is a polynomial time
algorithm for computing a SNE, see the appendix for the proof.

\begin{theorem}\label{thm:identical}
  Consider a bottleneck routing game on a directed graph $\mathcal{D}=(V,A)$
  with identical non-decreasing arc-cost functions and $n$ players
  having unit demand each that have to be routed from a common source
  to a common sink.  Then, there is a polynomial time algorithm
  computing a SNE.
\end{theorem}

To the best of our knowledge, our result establishes for the first
time an efficient algorithm computing a SNE in this setting. For the
next result, we assume the unit cost model of arithmetic operations,
that is, we assume that every arithmetic operation can be done in
constant time, regardless of the required precision. Furthermore, we
assume that cost functions on the facilities are bounded by a constant
$C\in\N$ that is polynomially bounded with respect to the input size.
Now, consider a bottleneck routing game $G$ on a directed graph
$\mathcal{D} = (V,A)$ with arc-costs $c$ and constant $C$. By scaling
every cost function of $G$ with the factor $1/C$, we obtain the
bottleneck routing game $\tilde{G}$. It is easy to see that SNE
coincide in both games.  Moreover, $c_a(x)/C \leq 1$ for all $a \in A$
and $x \in X$, thus, $(c_a(x)/C)^M\leq 1$ for every $M>0$. This
construction enables us to establish an efficient algorithm computing
SNE even for non-identical arc-costs (proof can be found in the
appendix).

\begin{proposition}
\label{prop:C}
Consider a class of bottleneck routing games on directed graphs
$\mathcal{D}=(V,A)$ with convex arc-cost functions $c_a,a\in A,$ that
are bounded by a constant $C$, and players having a unit demand each
that must be routed from a common source to a common sink.  Then,
there is a polynomial algorithm computing a SNE for every game in this
class.
\end{proposition}

Note that the SNE computed with the algorithm proposed in Proposition
\ref{prop:C} need not to coincide with the strict Pareto optimal
strategy as shown in Example~\ref{ex:msne} given in the appendix.  We
also illustrate in Example~\ref{ex:psne} structural differences
between singleton bottleneck congestion games and bottleneck routing
games.

\section{Infinite Strategic Games}

We now consider \emph{infinite} strategic games in which the players'
strategy sets are topological spaces and the private cost functions
are defined on the product topology.

Formally, an infinite game is a tuple $G=(N, X, \pi)$, where
$N=\{1,\dots,n\}$ is a set of players, and $X=X_1\times\dots\times
X_n$ is the set of pure strategies, where we assume that
$X_i\subseteq\R^{n_i},\, i\in N$ are compact sets, and $p=\sum_{i\in
  N}n_i$.  The cost function for player $i$ is defined by a
non-negative real-valued function $\pi_i:X\rightarrow \R_+,\;i\in N.$

As in the previous section, we are interested in conditions for
establishing existence of a generalized strong ordinal potential.
Unfortunately, even if an infinite game has the LIP,
Theorem~\ref{theorem:equivalent} (and in particular
Corollaries~\ref{cor:sne-ex}, \ref{cor:ord-pot},
and~\ref{cor:number-of-moves}) need not hold as they rely on the
existence of a strictly positive parameter $\epsilon$ that is a lower
bound on the minimal performance gain of a member of a coalition
performing an improving move. In infinite games, however, such a
constant need not exist since strategy sets are topological spaces and
the minimal performance gain may be unbounded from below.

We recall a famous result of Debreu \cite{Debreu:1954}, who showed
that the lexicographical ordering on an uncountable subset of
$\mathbb{R}^2$ cannot be represented by a real-valued function. It is
easy to derive that this also holds for the sorted lexicographical
ordering as defined in Definition
\ref{definition:lexicographic_sorting} as well. To see this, suppose
there is a real-valued function $\alpha$ representing the sorted
lexicographical order on $[0,1]\times[0,1]$, in particular $\alpha$
represents the sorted lexicographical ordering on $[2/3,1] \times
[0,1/3]$ where the sorted lexicographical order and the
lexicographical order coincide. Thus, we derive a contradiction.

This implies that a generalized strong ordinal potential need not
exist in general. Still, we are able to prove existence of SNE in
infinite games having the LIP for a \emph{continuous} function $\phi :
X \to \mathbb{R}^q_+$.

\begin{theorem}\label{theorem:continuous}
  Let $G$ be an infinite game that satisfies the LIP for a continuous
  function $\phi$. Then $G$ possesses a SNE.
\end{theorem}
The proof can be found in the appendix.  This result establishes the
existence of a SNE in all infinite games $G = (N,X,\pi)$ with compact
strategy spaces $X$ that have the LIP for a continuous function $\phi
: X \to \mathbb{R}_+^q$. Although the proof of Theorem
\ref{theorem:continuous} is constructive, the effort needed to compute
a SNE in such a game is very high as it involves the calculation of up
to $q!$ strategy profiles. We thus proceed by investigating some
special cases of infinite games possessing the LIP and identify cases
in which SNE can be computed efficiently.

\subsection{Infinite Bottleneck Congestion Games}
In this section, we introduce the "continuous counterpart" of bottleneck
congestion games. We are given a congestion model $\mathcal{M} = (N,
F, X, (c_f)_{f \in F})$ with $X_i=\{x_{i1},\dots,x_{in_i}\},\, n_i\in
\N,\, i\in N$, where as usual every $x_{ij}$ is a subset of facilities
of $F$.

From $\mathcal{M}$ we derive a corresponding \emph{infinite congestion
  model} $\mathcal{IM} = (N, F, X, d, \Delta, (c_f)_{f \in F})$, where
$d\in \R^n_+$, $\Delta=\Delta_1\times\cdots\times\Delta_n,$ and
$\Delta_i=\{\xi_i=(\xi_{i1},\dots,\xi_{i\,n_i}):\;
\sum_{j=1}^{n_i}\xi_{ij}=d_i,\, \xi_{ij}\geq 0,\, j=1,\dots,n_i\}$.
The strategy profile $\xi_i=(\xi_{i1},\dots,\xi_{i\,n_i})$ of player
$i$ can be interpreted as a distribution of non-negative
\emph{intensities} over the elements in $X_i$ satisfying
$\sum_{j=1}^{n_i}\xi_{ij}=d_i$ for $d_i\in \R_+, i\in N$.  Clearly,
$\Delta_i$ is a compact subset of $R^{n_i}_+$ for all $i\in N$. For a
profile $\xi=(\xi_1,\dots,\xi_n)$, we define the set of used
facilities of player $i$ as $ F_i(\xi)= \big\{ f\in F : \; \text{
  there exists }j\in \{1,\dots n_i\} \text{ with } f\in x_{ij} \text{
  and } \xi_{ij} >0\big\}.$ We define the \emph{load} of player $i$ on
$f$ under profile $\xi$ by $\xi_i^f=\sum_{x_{ij}\in X_i\,:\, f\in
  x_{ij}} \xi_{ij}$, $i\in N, f\in F$.  In contrast to finite
bottleneck congestion games, we assume that cost functions
$c_f:X\rightarrow\R_+$ only depend on the \emph{total load} defined as
$\ell_f(\xi)=\sum_{i\in N}\xi_{i}^f$ and are non-decreasing.

\begin{definition}[Infinite bottleneck congestion game]
\label{def:allocation_inf}
Let $\mathcal{IM} = (N, F, X,d, \Delta, (c_f)_{f \in F})$ be an
infinite congestion model derived from $\mathcal{M}$.  The
corresponding \emph{infinite bottleneck congestion game} is the
strategic infinite game $G(\mathcal{IM})~=~(N,\Delta, \pi)$, where
$\pi$ is defined as $\pi=\varprod_{ i \in N} \pi_i$ and $\pi_i(\xi) =
\max_{f\in F_i(\xi)} c_f\big( \ell_f(\xi) \big)$.
\end{definition}
Examples of such games are bottleneck routing game with splittable
demands.

We are now ready to prove that infinite bottleneck congestion games
have the $\pi$-LIP.
\begin{theorem}\label{thm:all-alpha}
  Let $G(\mathcal{IM})~=~(N,\Delta, \pi)$ be an infinite bottleneck
  congestion game. Then, $G(\mathcal{IM})$ has the LIP for the functions $\phi :
  \Delta \to \mathbb{R}_+^n$ and $\psi : X \to \mathbb{R}_+^{m\,n}$
  defined as
\begin{align*}
  \phi_i(\xi) = \pi_i(\xi), \quad\text{ for all } i \in N, &&
  \psi_{i,f}(\xi)= \begin{cases}
    c_f(\ell_f(\xi)), &\text{if $f\in F_i(\xi)$}\\
    0, &\text{else}
\end{cases} \quad\text{ for all }
  i\in N, f \in F.
\end{align*}
\end{theorem}
The proof uses similar arguments as the proof of Theorem~\ref{thm:all}
in the previous section except, that we now use the fact that cost
functions are load-dependent and non-decreasing.

We will now show that the above two functions $\phi$ and $\psi$ may be
discontinuous even if facility cost functions are continuous.
\begin{example}
  Let $G(\mathcal{IM})~=~(N,\Delta, \pi)$ be an infinite bottleneck
  congestion game with continuous cost functions $c$.  Then, the
  function $\phi$ and $\psi$ as defined in Theorem~\ref{thm:all-alpha}
  may be discontinuous on $\Delta$. Consider a bottleneck congestion game
  with one player having two facilities $\Delta_1=\{\{r_1\},\{r_2\}\}$
  over which she has to assign a demand of size $1$. The facility
  $r_1$ has a cost function equal to the load, while facility $r_2$
  has a constant cost function equal to $2$. Let
  $\xi^2(\epsilon)=\epsilon>0$ be assigned on facility $r_2$ and the
  remaining demand $\xi^1(\epsilon)=1-\epsilon$ be assigned to $r_1$.
  Then, for \emph{any} $\epsilon>0$ we have $\phi(\xi(\epsilon))=2$,
  while $\phi(\xi(0))=1$.
\end{example}
By assuming that cost functions are continuous and \emph{strictly}
increasing, we obtain, however, the LIP for a \emph{continuous}
function $\nu$.
\begin{theorem}\label{thm:all-increasing}
  Let $G(\mathcal{IM})~=~(N,\Delta, \pi)$ be an infinite bottleneck
  congestion game with strictly increasing cost functions. Then,
  $G(\mathcal{IM})$ has the LIP for the function $\nu : \Delta \to
  \mathbb{R}_+^{m}$ defined as $ \nu_{f}(\xi)= c_f(\ell_f(\xi))$ for
  all $f \in F.$
\end{theorem}
We prove the theorem in the appendix.

This result relies on the strict monotonicity of $c_f$ and cannot be
generalized to non-strict monotonic functions. We illustrate this
issue with the following example. Consider a bottleneck congestion
game with two players having equal demands and two facilities $F =
\{f,g\}$, where $c_f(\ell) = 10$ and $c_g(\ell) = \ell$. Both players
may either choose $f$ or $g$. Then $\bigl((f,f),(g,f)\bigr)$ is an
improving move for player $1$ as his private cost decreases from $10$
to $1$. However $\nu\bigl((f,f)\bigr) = (10,0)$ and
$\nu\bigl((g,f)\bigr) = (10,1)$ and hence $\nu$ is not
lexicographically decreasing along this improving move. This example
shows that strict monotonicity is essential in the proof of the above
theorem.

We overcome this problem by slightly generalizing the notion of
lexicographical ordering to ordered sets that are different from
$(\mathbb{R},\leq)$. To this end, consider a totally ordered set
$(\mathcal{A},\leq_\mathcal{A})$. Similar to
Definition~\ref{definition:lexicographic_sorting}, we introduce a
lexicographical order on $\mathcal{A}$-valued vectors. For two vectors
$a,b \in \mathcal{A}^q$, let $\tilde{a}$ and $\tilde{b}$ be two
vectors that arise from of $a$ and $b$ by ordering them w.r.t
$\leq_\mathcal{A}$ in non-increasing order. We say that $a$ is
$\mathcal{A}$-lexicographically smaller than $b$, written $a
\prec_\mathcal{A} b$ if there is $m \in \{1,\dots,q\}$ such that
$\tilde{a}_i =_{\mathcal{A}} \tilde{b}_i$ for all $i< m$ and
$\tilde{a}_m <_\mathcal{A} \tilde{b}_m$.

Consequently a strategic game satisfies the $\mathcal{A}$-LIP if there
are $q \in \mathbb{N}$ and a function $\phi : X \to \mathcal{A}^q$
such that $\phi(x) \succ_\mathcal{A} \phi(y)$ for all $(x,y) \in I$.

The following theorem establishes the $\mathcal{A}$-LIP for infinite
bottleneck congestion games, where $(\mathcal{A},\leq_\mathcal{A}) =
(\mathbb{R}^2, \leq_{lex})$ and $\leq_{lex}$ denotes the ordinary
lexicographical order (that does not involve any sorting of the entries)
on $\mathbb{R}^2$, that is, $(a_1,a_2) <_{lex} (b_1,b_2)$ if either
$a_1 < b_1$ or $\bigl(a_1= b_1 \text{ and } b_2 < b_2\bigr)$.

\begin{theorem}\label{thm-A-LIP}
  Let $(\mathcal{A},\leq_\mathcal{A}) = (\mathbb{R}^2, \leq_{lex})$
  and let $G(\mathcal{IM})~=~(N,\Delta, \pi)$ be an infinite
  bottleneck congestion game. Then, $G(\mathcal{IM})$ has the
  $\mathcal{A}$-LIP for $\phi : \Delta \to \mathcal{A}^{m}$ defined as
  $ \phi_{f}(\xi)= \bigl(c_f(\ell_f(\xi)), \ell_f(\xi)\bigr)$ for all
  $f \in F.$
\end{theorem}

The proof is similar to that of Theorem \ref{thm:all-increasing}. Let
$(\xi,\nu)$ be an improving move. Again, we denote by $g$ one of the
bottleneck facilities of that player of the coalition with highest
cost before the improving move. Instead of the strict monotonicity of
$c_g$ we use that in the case $c_g(\ell(\xi)) = c_g(\ell(\nu))$, we
still get $\ell(\xi) > \ell(\nu)$. Thus, the definition of
$\leq_{lex}$ implies $\bigl(c_g(\ell(\xi)),\ell(\xi)\bigl) \succ
\bigl(c_g(\ell(\nu)),\ell(\nu)\bigr)$, proving the result.
\begin{corollary}
  Let $G(\mathcal{IM})~=~(N,\Delta, \pi)$ be an infinite bottleneck
  congestion game with continuous  cost
  functions.  Then, $G(\mathcal{IM})$ possesses an SNE.
\end{corollary}
Consider the function $\phi$ as in Theorem~\ref{thm-A-LIP}.  We
observe that $\phi$ is continuous as $(c_f)_{f\in F}$ is continuous.
Thus, a similar (but slightly more involved) proof as that of
Theorem~\ref{theorem:continuous} implies that there is strategy
profile $\xi\in\Delta$ (where $\Delta$ is compact) that minimizes
$\phi$ w.r.t. the order defined on $\mathcal{A}$.

The above result establishes for the first time the existence of SNE
for a variety of games such as scheduling games with malleable jobs,
bottleneck routing games with splittable demands, etc. Note that this
result gives also an alternative and constructive proof for the
existence of PNE in bottleneck routing games with splittable demands
compared to the proof of Banner and Orda~\cite{BannerO07}.

\subsection{Approximate SNE}

We introduce the notion of \emph{$\alpha$-approximate} strong Nash
equilibria for infinite games. We denote by $I^{\alpha}(S)\subset X
\times X $ the set of tuples $(x, (y_S,x_{-S}))$ of $\alpha$-improving
moves for $S\subseteq N$ and define by $I^{\alpha}$ their union. A
strategy profile $x$ is an $\alpha$-approximate strong Nash
equilibrium if no coalition $\emptyset\neq S\subseteq N$ has an
alternative strategy profile $y_S$ such that $
\pi_i(x)-\pi_i(y_S,x_{-S})> \alpha, \text{ for all }i\in S.$ We call a
function $P : X \to \mathbb{R}$ an $\alpha$-generalized strong ordinal
potential if $(x,y)\in I^{\alpha}$ implies $ P(x) - P(y) < 0.$ We also
define the $\alpha$-lexicographical improvement property ($\alpha$-LIP)
and $\alpha$-$\pi$-LIP similar to Definition~\ref{definition:lip},
except that we replace $I$ with $I^{\alpha}$, respectively.

We will prove in the appendix that bottleneck congestion games with
bounded cost functions possess an $\alpha$-approximate SNE for every
$\alpha>0$.
\begin{theorem}\label{thm:approximate}
  Let $G$ be an infinite bottleneck congestion game with bounded cost
  functions. Then, $G$ possesses an $\alpha$-approximate SNE for every
  $\alpha>0$.
\end{theorem}

Note that there is a fundamental difference to the result of Monderer
and Shapley~\cite{Monderer:1996}, who showed that every infinite game
having an \emph{exact} bounded potential (on a compact strategy space)
possesses an $\alpha$-PNE for every $\alpha>0$.  Monderer and Shapley
use in their proof that the payoff difference of a deviating player is
equal to the potential difference.  This, however, is not true for
generalized ordinal potentials and, in fact, we crucially exploit the
combinatorial structure of infinite bottleneck congestion games in the
proof of Theorem~\ref{thm:approximate}.

We finally show that there is a polynomial algorithm computing an
$\alpha$-approximate SNE for every $\alpha>0$ for the class of
bottleneck routing games with splittable flow and bounded convex cost
functions, where the upper bound is polynomial in the input size.
\begin{proposition}
  Consider a class of splittable bottleneck routing games on
  multi-commodity graphs $\mathcal{D}=(V,A)$ with bounded convex
  arc-cost functions, that are bounded by a constant $C$, and players
  having arbitrary positive demands. Then, there is a polynomial
  algorithm computing an $\alpha$-SNE for arbitrary $\alpha>0$.
\end{proposition}
\begin{proof}
  The $\alpha$-generalized strong ordinal potential function
  $P_{M(\alpha)}$ as defined in Lemma~\ref{lem:approximate} is convex
  (as $c$ is convex) and the flow polytope is compact.  Thus, one can
  apply the same scaling argument as in Proposition~\ref{prop:C} to
  obtain a polynomial time algorithm (e.g., the ellipsoid method)
  computing a splittable flow $z$ that satisfies $P_{M(\alpha)}(z)\leq
  \min_{\xi\in \Delta} P_{M(\alpha)}(\xi)-\epsilon$ for arbitrary
  $\epsilon>0$.  \qed\end{proof}

\section{Extensions}
We present two extensions that (as we feel) are the most interesting ones.

Recently, Rozenfeld \cite{Rozenfeld:2007} introduced an even stronger
solution concept than SNE.  In a \emph{super strong equilibrium},
mnemonic SSNE, no coalition has an alternative strategy that is
profitable to \emph{one} of its members while being at least neutral
to all other members. Formally, a strategy profile $x$ is a super
strong Nash equilibrium if there is no coalition $\emptyset \neq
S\subseteq N$ such that there is an alternative strategy profile $y_S
\in X_S$ with $\pi_i(y_S,x_{-S})-\pi_i(x)\leq 0$ for all $i\in S$,
where the inequality is strict for at least on player in $S$.
Clearly, every SSNE is also a SNE but the converse need not hold.
However, one can generalize the concept of LIP to also allow for the
enlarged set of improving moves. It follows that all our results
regarding SNE for finite and infinite bottleneck congestion games
carry over to SSNE.

A natural generalization of bottleneck congestion games can be
obtained by assuming that players are \emph{heterogeneous} with
respect to the cost of the most expensive facility, that is, they
attach different \emph{values} to the cost of the most expensive
facility.  We can model this heterogeinity by introducing a
player-specific function $\varpi_i : \mathbb{R}_+ \to \mathbb{R}_+$
that maps the cost of a facility to the private cost experienced by
the player. Assuming that higher costs on facilities are associated
with higher private costs, that is, $\varpi_i$ is strictly increasing
for all $i \in N$, we can actually show that these games possess the
LIP (though not the $\pi$-LIP).  \bibliographystyle{abbrv}
\bibliography{Bib}

\newpage
\appendix
\section*{Appendix}

\subsection*{Proof of Theorem~\ref{theorem:equivalent}}
\begin{proof}
  We prove $1.\Rightarrow 2.\Rightarrow 3.\Rightarrow 4.\Rightarrow
  5.\Rightarrow 1.$

  $1.\Rightarrow 2.$ is trivial.

  $2.\Rightarrow 3.$: We use that every directed acyclic graph
  possesses a topological order, which gives rise to a generalized
  ordinal potential by simply assigning real-valued labels to strategy profiles
  according to their topological order.
 
  $3.\Rightarrow 4.$: Let $Q$ be a generalized ordinal potential
  function and let $Q_{\min}=\min_{x\in X}Q(x)$ . We define
  $\phi(x)=Q(x)+Q_{\min}$ and it follows that $G$ together with $\phi$
  has the LIP.
  
  $4.\Rightarrow 5.$: Let $(x,(y_S,x_{-S})) \in I(S), S\subseteq N$ be
  arbitrary and let $\phi$ be as in the definition of LIP.  We will
  show that there is a constant $M > 0$ such that $P(x) -
  P(y_S,x_{-S}) = \sum_{i \in N} \phi_i(x)^{M'} -
  \phi_i(y_S,x_{-S})^{M'} > 0$ for all $M < M'$. To this end, we denote
  by $\tilde{\phi}(x)$ and $\tilde{\phi}(y_S,x_{-S})$ the vectors that
  arise by sorting $\phi(x)$ and $\phi(y_S,x_{-S})$ in non-increasing
  order. As $\phi(y_S,x_{-S}) \prec \phi(x)$, there is an index $m \in
  \{1,\dots,q\}$ such that $\tilde{\phi}_i(x) =
  \tilde{\phi}_i(y_S,x_{-S})$ for all $i < m$ and $\tilde{\phi}_m(x) <
  \tilde{\phi}_m(y_S,x_{-S})$. We then obtain
\begin{align}
  P_{M'}(x) - P_{M'}(y_S,x_{-S})&= \sum_{i=1}^q \phi_i(x)^{M'} - \sum_{i=1}^q\phi_i(y_S,x_{-S})^{M'} \notag\\
  &= \tilde{\phi}_m(x)^{M'} - \tilde{\phi}_m(y_S,x_{-S})^{M'} + \sum_{i=m+1}^q \tilde{\phi}_i(x)^{M'} - \sum_{i=m+1}^q\tilde{\phi}_i(y_S,x_{-S})^{M'} \notag\\
  &\geq \tilde{\phi}_m(x)^{M'} - \tilde{\phi}_m(y_S,x_{-S})^{M'}  - (q-m)\tilde{\phi}_m(y_S,x_{-S})^{M'}\notag\\
  &\geq \tilde{\phi}_m(x)^{M'} - q\tilde{\phi}_m(y_S,x_{-S})^{M'}.
  \label{equation:last_line}
\end{align}
Standard calculus shows that the expression on the right hand side
of~\eqref{equation:last_line} is positive if 
\[M'~>~\log(q) \,/\, \bigl(\log(\tilde{\phi}_m(x)) -
\log(\tilde{\phi}_m(y_S,x_{-S}))\bigr) > 0.\] Clearly, $M'$ depends on
$(x,y) \in I$, but as the number of improvement steps is finite, we
may chose $M := \max_{(x,y) \in I} M'\bigl((x,y)\bigr)$ and get the
claimed result.

$5.\Rightarrow 1.$ is trivial.  
 \qed\end{proof}

\subsection*{Proof of Corollary~\ref{cor:ord-pot}}
\begin{proof}
  In the proof of Theorem \ref{theorem:equivalent} we need $M >
  \log(q) \,/\, \bigl(\log(\tilde{\phi}_m(x)) -
  \log(\tilde{\phi}_m(y))\bigr)$ for every improving move
  $(x,y) \in I$. Here, $m$ is the first index such that
  $\tilde{\phi}_m(x) > \tilde{\phi}_m(y)$.  The mean value
  theorem implies that $\log(\tilde{\phi}_m(x)) -
  \log(\tilde{\phi}_m(y)) = \bigl(\tilde{\phi}_m(x) -
  \tilde{\phi}_m(y)\bigr)\, / \, \xi$ for some $\xi \in
  \bigl(\tilde{\phi}_m(x), \tilde{\phi}_m(y)\bigr)$ and hence
\begin{align*}
  M > \log(q) \frac{\phi_{\max}}{\epsilon_{\min}} \geq \frac{\log(q)
    \,\tilde{\phi}_m(x)}{\tilde{\phi}_m(x) -
    \tilde{\phi}_m(y)} \geq
  \frac{\log(q)}{\log(\tilde{\phi}_m(x)) -
    \log(\tilde{\phi}_m(y))},
\end{align*}
establishing the result.
\qed\end{proof}
\subsection*{Proof of Theorem~\ref{thm:stab}}
\begin{proof}
  Let $p, q \in \mathbb{N}$ and $p < q$. As the $L_p$-norm is
  decreasing in $p$ and by H\"older's inequality, we get $L_q(x) \leq
  L_p(x) \leq \sqrt[q]{n^{q-p}} L_q(x)$, and for the special case $q =
  \infty$ we get $L_\infty(x) \leq L_p(x) \leq \sqrt[p]{n}
  L_\infty(x)$. The latter inequality implies that $\lim_{p \to
    \infty} L_p(x) = L_\infty(x)$.
    
  As $G$ satisfies the $\pi$-LIP, $P_{M}(x) = \sum_{i \in N} \pi(x)^M$
  is a generalized strong ordinal potential of $G$ for all $M$ large
  enough and hence the minimum $x^*$ of $P_M$ over $X$ is a SNE. As
  the $M$th root is a monotone function, $x^*$ minimizes $L_M(x) =
  \sqrt[M]{\sum_{i \in N} \pi(x)^M}$ as well.  As $X$ is finite, we
  can choose $M$ large enough so that the strategy profile $x^*$
  minimizes $L_\infty(x)$.

  For proving the second claim, let $x^*$ be a strong Nash equilibrium
  minimizing the generalized strong ordinal potential $P_M$ and
  $L_M(x) = \sqrt[M]{\sum_{i \in N} \pi(x)^M}$ and let $y$ be a
  strategy profile minimizing $L_p$. Then, we derive the inequalities
  $ L_p(x^*) \leq \sqrt[M]{n^{M-p}}L_M(x^*) \leq \sqrt[M]{n^{M-p}}
  L_M(y) \leq \sqrt[M]{n^{M-p}} L_p(y).  $ The second inequality is
  valid as $x^*$ is a potential minimizer. Thus, $L_p(x^*) \leq n^{1 -
    p/M}L_p(y) < nL_p(y), $ which proves the claimed result.
  \qed\end{proof} In the following we provide an example of a class of
games with the $\pi$-LIP whose parameters can be chosen in a way such
that the price of stability w.r.t. $L_p$ is arbitrarily close to
$\sqrt[p]{n}$, implying that the result of Theorem \ref{thm:stab}
w.r.t $L_1$ is tight.

\begin{example}[Price of stability]
\label{example:root}
We consider the game $G = (N,X,\pi)$ with $N = \{1,\dots,n\}$ with $X_1 = X_2 = \{0,1\}$ and $X_i = \{0\}$ for $3\leq i\leq n$. Private costs are shown in Fig.~\ref{figure:examples}a.

\begin{figure}[bp]
\centering
\begin{minipage}{0.45\linewidth}
\centering
\hspace{1cm} $0$ \hspace{2.5cm} $1$ \\[3pt]
\begin{tabular}{c | c | c |}
\cline{2-3}
$0$~~ 	& ~$(k-\epsilon,k-\epsilon, k - \epsilon, \dots, k - \epsilon)$~ & ~$(k,k, k\dots, k)$~ \\ 
\cline{2-3}
$1$~~ 	& ~$(k,0,0,\dots,0)$~ & ~$(k,\epsilon,k, \dots, k)$~\\
\cline{2-3}
\end{tabular}\\[6pt]
\textbf{a)}
\end{minipage}
\hspace{0.05\linewidth}
\begin{minipage}{0.45\linewidth}
\centering
\hspace{0.4cm} $0$ \hspace{0.5cm} $1$ \\[3pt]
\begin{tabular}{c | c | c |}
\cline{2-3} 
$0$~~ & ~$(0,0)$~ & ~$(0,k)$~\\ 
$1$~~ & ~$(k,k)$~ & ~$(0,k)$~\\
\cline{2-3}
\end{tabular}\\[6pt]
\textbf{b)}
\end{minipage}
\caption{\textbf{a)} Private costs received by the players for
  strategy profiles $X_1 \times X_2$ of the game considered in Example
  \ref{example:root}. \textbf{b)} A game with unbounded price of
  anarchy w.r.t. any $L_p$-norm.}
  \label{figure:examples}

\end{figure}
It is straightforward to check that this game has the $\pi$-LIP. The
unique SNE is the strategy profile $(0,\dots,0)$ realizing a private cost vector of
$(k-\epsilon, \dots, k-\epsilon)$. For any $p \in \mathbb{N}$, there
is $\epsilon > 0$ such that $L_p(\cdot)$ is maximized in strategy profiles $(1,0,0,\dots,0)$ realizing a cost vector of $(k,0,\dots,0)$. Hence the price of stability approaches
$\sqrt[p]{n}$ arbitrarily close.
\end{example}

So far, our results concern the price of stability only. The next
example shows that games with the $\pi$-LIP may have a price of anarchy
that is unbounded.

\begin{example}[Unbounded price of anarchy]
  Consider the game $G = (N,X,\pi)$ with $N = \{1,2\}$, $X_1 = X_2 =
  \{0,1\}$ and private costs given in Fig.~\ref{figure:examples}b for
  any $k>0$. It is straightforward to check that this game has the
  $\pi$-LIP and that both $(0,0)$ and $(1,1)$ are SNE.  Hence, the
  price of anarchy w.r.t.\ any $L_p$ norm is unbounded from above.
\end{example}
\subsection*{Proof of Corollary~\ref{cor:min-max}}
We establish the result by proving the following characterization:
\begin{lemma}
  Let $G$ be a finite strategic game having the $\pi$-LIP.  Then, a
  strategy profile $x$ minimizes $P_{M}$ as defined in
  Theorem~\ref{theorem:equivalent} with $\phi=\pi$ if and only if $x$
  is min-max fair.
\end{lemma}
\begin{proof}
  $"\Rightarrow":$ Let $x$ minimize $P_{M}$ as defined in
  Theorem~\ref{theorem:equivalent} with $\phi=\pi$. Assume by
  contradiction that there is another strategy profile $y$ such that
  $\pi_i(y)<\pi_i(x)$ and $\pi_j(y)\leq \pi_j(x)$ for all $j\in N$
  with $\pi_j(x)\geq \pi_i(x)$. Then, $P_M(x)-P_M(y)>0$ for $M$ large
  enough, contradicting the minimality of $x$.

  $"\Leftarrow":$ Let $x$ be min-max fair. Assume by contradiction
  that $x$ is not a minimizer of $P_M$. This implies that there exists
  $y\in X$ with $\pi(x)\succ \pi(y)$. Thus, there exists an index $m$
  such that $\tilde{\pi}(x)_i=\tilde{\pi}(y)_i$ for all $i< m$, and
  $\tilde{\pi}(x)_m>\tilde{\pi}(y)_m$. This, however, implies that $x$
  is not min-max fair. 
\qed\end{proof}
\subsection*{Proof of Theorem~\ref{thm:all}}
\begin{proof}
  We first prove the claim for $\psi$. Consider an improving move
  $(x,(y_S,x_{-S}))\in I$. Let $j\in S$ be a member of the coalition
  with highest cost before the improvement step, i.e.,
  $j~\in~\arg\max_{i \in S} \pi_i(x)$. We set $\Psi^+ := \{(i,f) \in N
  \times F : \psi_{i,f}(x) \geq \pi_j(x)\}$ and claim that
  $\psi_{i,f}(x) \geq \psi_{i,f}(y_S,x_{-S})$ for all
  $(i,f)\in\Psi^+$.  To see this, suppose there is $(k,g) \in \Psi^+$
  such that $\psi_{k,g}(x) < \psi_{k,g}(y_S,x_{-S})$.  The
  independence of irrelevant choices and the monotonicity of cost
  functions imply that a member $i\in S$ of the coalition uses $g$. So
\begin{align*}
\pi_j(x) \geq \pi_i(x) > \pi_i(y_S,x_{-S}) \geq \psi_{k,g}(y_S,x_{-S}),
\end{align*}
which contradicts $(k,g)\in\Psi^+$. 

Now we define $\Psi^- := \{(i,f) \in N \times F : \psi_{i,f}(x) <
\pi_j(x)\}$ and claim that $\psi_{i,f}(y_S,x_{-S}) < \pi_j(x)$ for all
$(i,f)\in\Psi^-$. To see this, suppose there is $(k,g) \in \Psi^-$
such that $\psi_{k,g}(y_S,x_{-S}) \geq \pi_j(x)$.  Because of the
monotonicity of the cost functions and the independence of irrelevant
choices, there is a member $i\in S$ of the coalition using $g$ in
$(y_S,x_{-S})$ giving rise to
\begin{align*}
\pi_j(x) \geq \pi_i(x) > \pi_i(y_S,x_{-S}) \geq \psi_{k,g}(y_S,x_{-S}) \geq \pi_j(x),
\end{align*}
which is a contradiction.

We remark that $N \times F = \Psi^+ \cup \Psi^-$ and that we have shown
that $\psi_{i,f}(x) - \psi_{i,f}(y_S,x_{-s})\geq 0$ for all $(i,f) \in
\Psi_{S}^+$ and $\psi_{i,f}(y_S,x_{-S})\leq \psi_{k,g}(y_S,x_{-S})$
for all $(i,f) \in \Psi^-$ and $(k,g)\in\Psi_{S}^+$. Since $j\in S$
and $\pi_j(x)>\pi_j(y_S,x_{-S})$, there exists $(j,f)\in\Psi^+$ with
$\psi_{j,f}(x) > \psi_{j,f}(y_S,x_{-s})$. Hence, $\psi(x) \succ
\psi(y_S,x_{-S})$, finishing the first part of the proof.

To prove the LIP for $\phi$, we must show that $(\pi_i)_{i\in N}
\succ (\pi_i(y_S,x_{-S}))_{i \in N}$ for all $(x,(y_s,x_{-S})) \in I$.
We again consider a player $j~\in~\arg\max_{i \in S} \pi_i(x)$ and
decompose the set of players into the sets
\begin{align*}
  N^+ := \{i \in -S : \pi_i(x) \geq \pi_j(x)\} && N^-:= \{i \in N :
  \pi_i(x) < \pi_j(x)\}.
\end{align*}
A similar argument as in the first part shows that $\pi_i(x) \geq
\pi_i(y_S,x_{-S})$ for all $i \in N^+$ and that $\pi_i(y_S,x_{-S}) <
\pi_j(x)$ for all $i \in N^-$, establishing the result.
  \qed\end{proof}

\subsection*{Proof of Theorem~\ref{thm:identical}}
\begin{proof}
  We assign a uniform capacity of $1$ to each arc $a \in A$.  With the
  polynomial algorithm of Edmonds and Karp we obtain both a minimum
  $(s,t)$-cut $C$, and a maximum flow $x$. Let us say that $C :=
  \{a_1,\dots,a_m\}$ for some $m \in \mathbb{N}$. As we may assume
  without loss of generality that $x$ is integer and all capacities
  are $1$, we may decompose the flow $x$ into $m$ arc-disjoint paths
  $\mathcal{P}_j$.
  
  We set $k=m\,\left\lceil\frac{n}{m}\right\rceil-n$ and consider the
  strategy profile $x$ in which $\left\lfloor\frac{n}{m}\right\rfloor$
  players are routed along each path $P_j, j=1,\dots,k,$ and
  $\left\lceil\frac{n}{m}\right\rceil$ players are routed along each
  of the other paths. As all arcs have the same arc-cost function, we
  derive that $\pi_i(x) = \max_{a \in x_i} = c_a(x) = c_{a_j}(x)$ for
  some $1 \leq j \leq m$. Since every other flow traverses the cut
  $C$, there can be no deviation of a coalition $S$ that is profitable
  to all of its members, that is, strictly reduces their private
  costs.
\qed\end{proof}

\subsection*{Proof of Proposition~\ref{prop:C}}
\begin{proof}
  Instead of computing a SNE in the original game, we consider the
  game $\tilde{G}$ with arc-costs $\tilde{c}_a(x) = c_a(x)/C$.
  Obviously, each strategy profile $x\in X$ establishes an integral
  $(s,t)$-flow with value $n$.  Conversely, each such flow can be
  decomposed into $n$ paths starting in $s$ and ending in $t$, see
  \cite{AMO93}. So, there is a one-to-one correspondence between
  strategy profiles and integral $(s,t)$-flows with value $n$.

  Let $n_a(x)$ denote the number of players using facility $a$ under
  strategy profile $x$. As the bottleneck routing game has the
  $\pi$-LIP for the function $\psi$ defined in Theorem~\ref{thm:all},
  the function $P_M(x) := \sum_{a \in A} \tilde{c}_a(x)^M\,n_a(x)$ is
  a generalized strong potential function of the bottleneck routing
  game, see the construction of the generalized ordinal potential in
  Theorem~\ref{theorem:equivalent}. We can rewrite the potential using
  flow variables $x:\R^{|A|}\rightarrow \R_+$ and obtain $P_M(x) :=
  \sum_{a \in A} \tilde{c}_a(x_a)^M\,x_a$, where $x_a$ denotes the
  total flow on arc $a$. Moreover, note that the optimization problem
  $\min_{x \in X} P_M(x)$ of computing a minimal integral flow with
  convex arc-cost can be solved in polynomial time (given $P_M(x)$),
  see Ahuja et al.~\cite{AMO93}. The optimal solution $x^*$ of this
  problem minimizes the generalized strong potential function $P_M$
  and hence is a SNE.  \qed\end{proof}
\subsection*{Examples~\ref{ex:msne} and~\ref{ex:psne}}
\begin{example}\label{ex:msne}
  Consider the symmetric bottleneck routing game with players set $N
  =\{1,\dots,n\}$ depicted in Fig.~\ref{figure:multiple_sne}. The
  strategy set $X_i$ of each player $i \in N$ comprises all paths from
  $s$ to $t$, that are $P_1 := \{(sa),(at)\}, P_2 := \{(sb),(bt)\}$
  and $P_3 := \{(sb),(ba),(bt)\}$. In addition, we consider the
   cost functions
\begin{align*}
c_1(\ell) =
\begin{cases}
0 &\text{ if } \ell < n\\
1 &\text{ else },
\end{cases} && c_2(\ell) =
\begin{cases}
1 &\text{ if } \ell \leq 1\\
2 &\text{ else }.
\end{cases}
\end{align*}
and the graph depicted in Fig.~\ref{figure:multiple_sne}a.

\begin{figure}[tpb]
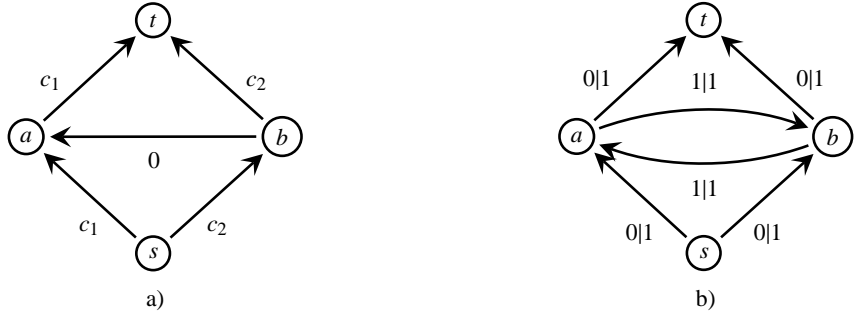

\begin{minipage}{0.45\linewidth}
\begin{center}
\psset{arrowsize=8pt,arrowlength=1,linewidth=1pt,nodesep=2pt,shortput=tablr}
\begin{psmatrix}[colsep=12mm,rowsep=10mm, mnode=circle]
		&$t$			\\
$a$	&		&$b$	\\
		&$s$		
\ncline{->}{3,2}{2,1}_{$c_1$}
\ncline{->}{2,1}{1,2}<{$c_1$}
\ncline{->}{3,2}{2,3}_{$c_2$}
\ncline{->}{2,3}{1,2}>{$c_2$}
\ncline{->}{2,3}{2,1}_{$0$}
\end{psmatrix}
\\[6pt]
a)
\end{center}

\end{minipage}
\begin{minipage}{0.45\linewidth}
\begin{center}
\psset{arrowsize=8pt,arrowlength=1,linewidth=1pt,nodesep=2pt,shortput=tablr}
\begin{psmatrix}[colsep=12mm,rowsep=10mm, mnode=circle]
 		&$t$			\\
$a$	&		&$b$	\\
		&$s$		
\ncline{->}{3,2}{2,1}_{$0|1$}
\ncline{->}{2,1}{1,2}<{$0|1$}
\ncline{->}{3,2}{2,3}_{$0|1$}
\ncline{->}{2,3}{1,2}>{$0|1$}
\ncarc[arcangle=20]{->}{2,3}{2,1}_{$1|1$}
\ncarc[arcangle=20]{->}{2,1}{2,3}^{$1|1$}
\end{psmatrix}
\\[6pt]
b)
\end{center}
\end{minipage}
\caption{a) Bottleneck routing game with multiple SNE b) Bottleneck routing game with non-strong PNE}
\label{figure:multiple_sne}
\end{figure}

There are two types of SNE.
In the first type, $n-1$ players play $P_1$ and one player plays
$P_2$. The players on $P_1$ experience a cost of $0$ while the single
player on $P_2$ experiences a cost of $1$. Thus, the sum of all costs
equals $1$. In the other type of SNE, again $n-1$ players play $P_1$
while one player plays $P_3$. Hence all players experience a cost of
$1$ and the total cost of all players sums up to $n$. The two types of
SNE actually correspond to the global minima of the generalized strong
ordinal potential functions derived from the $\pi$-LIP w.r.t. $\phi$
and $\psi$ as defined in Theorem~\ref{thm:all}.
\end{example}
\begin{example}\label{ex:psne}
  Consider the instance in Fig.~\ref{figure:multiple_sne}\,b and
  assume there are two players with unit demand each. One can easily
  see that this instance admits a PNE (routing both demands along the
  two zig-zag paths) that is not a SNE.  This example contrasts a
  result of Holzman and Law-Yone~\cite{Holzman97}, who have shown that
  for singleton congestion games, every PNE is also a SNE. Thus,
  allowing more complex strategies (paths instead of single
  facilities) makes a structural difference.
\end{example}
\subsection*{Proof of Theorem~\ref{theorem:continuous}}
\begin{proof}
  By assumption there exists $q\in \N$ and a function
  $\phi:X\rightarrow \R^q_+$ such that $\phi(y_S,x_{-S}) \prec \phi(x)$ for all $(x, (y_S,x_{-S}))\in
  I,$.  We will show
  that there exists $x_{\min}\in X$ with $\phi(x_{\min}) \preceq
  \phi(y)$ for all $x,y\in X.$ Our proof is constructive and proceeds
  in $q$ phases. In the first phase, we solve the following program
  \begin{align} \tag{$P_1$}
 \min_{x \in X} \alpha\; s.t.: \; \phi_i(x)\leq \alpha, \;\text{ for all }\; i\in \{1,\dots,q\}.
\end{align}
Note that continuity of $\phi$ implies that the half-space $H_1 =
\{x\in X : \phi_i(x)\leq \alpha\}$ is compact. To see this, observe that
$H_1 = \phi_i^{-1}((-\infty,\alpha]) =
(\phi_i^{-1}((\alpha,\infty))^\text{C}$, where $\phi^{-1}$ denotes the
pre-image of $\phi$ and $(A)^\text{C}$ denotes the complement of $A$
with respect to $X$. As $\phi$ is continuous and $(\alpha,\infty)$ is
open $\phi_i^{-1}((\alpha,\infty))$ is open and hence $H_1$ is
closed.  Hence, $H_1$ is a closed subset of the compact set $X$ and
thus compact. As the objective of~$(P_1)$ is continuous, the minimum
is attained in $H_1$ with value $\alpha_{1}\in\R_+$. Let
$A_{1}\subseteq X$ denote the set of optimal solutions and let
$B_{1}=\{j\in \{1,\dots,q\}: \text{there exists } x^* \in A_{1} \text{
  with }\phi_j(x^*)=\alpha_{1}\}$ denote the set of indices for which
the optimal value is attained.  Clearly, $B_1$ is non-empty. If
$B_1$ contains a single element, say $j$, the
lexicographical minimum $x_{\min}$ fulfills $\phi_j(x_{\min}) =
\alpha_1$ and we can proceed by solving
\begin{align}\tag{$P_2^j$}\; \min_{x \in X} \alpha \;
  \text{ s.t.: } \phi_j(x)= \alpha_1,\;\phi_i(x)\leq \alpha, \;\text{ for all } \;
  i\in \{1,\dots,q\}\setminus\{j\}.
\end{align}
If, in contrast, $B_1$ contains more than one element, we solve
program $P_2^j$ for every $j \in B_1$. Continuing this way we obtain
at most $q!$ different solution vectors $\phi\in\R^q_+$. Taking the
lexicographically smallest among them, we obtain $x_{\min}$, which is a
SNE.  \qed\end{proof}
\subsection*{Proof of Theorem~\ref{thm:all-increasing}}
\begin{proof}
  We choose a deviating player $ j \in \arg\max_{i \in S} \pi_i(x)$
  with highest cost before the improving move and one of the
  facilities $g \in \arg\max_{f \in x_j} c_f(x)$ at which $\pi_j(x)$ is
  attained. Decompose $F$ into $F^+$ and $F^-$ defined as $ F^+ := \{
  f \in F : c_f(x) \geq c_g\}$ and $F^- := \{f \in F : c_f(x) <
  c_g\}$. We claim that $c_f(y_S,x_{-S}) \leq c_f(x)$ for all $f \in
  F^+$ and that $c_f(y_S,x_{-S}) < c_g(x)$ for all $f \in F^-$, which
  establishes the result with similar arguments as in the proof of
  Theorem \ref{thm:all}. Note that we use here that $c_g$ is strictly
  decreasing as player $j$ changes her strategy profile from $x_j$ to
  $y_j$.
\qed\end{proof}

\subsection*{Proof of Theorem~\ref{thm:approximate}}
We prove the theorem by stating a useful lemma.
\begin{lemma}\label{lem:approximate}
  Let the function $\psi : \Delta \to \mathbb{R}_+^{m\,n}$ be defined
  as 
\[ \psi_{i,f}(\xi)= \begin{cases}
    c_f(\ell_f(\xi)), &\text{if $f\in F_i(\xi)$}\\
    0, &\text{else}
  \end{cases} \quad\text{ for all } i\in N, f \in F.
\] 
Moreover, let $\alpha>0$ and define $P_{M(\alpha)}(\xi):=\sum_{f\in F,
  i\in N} \psi_{i,f}(\xi)^{M (\alpha)}$, where $M(\alpha)\geq
(2\,\psi_{\max}/\alpha+1)\,\log(n\,m)$ and $ \psi_{\max} := \sup_{\xi
  \in X, 1 \leq f \leq m} c_f(\ell_f(\xi)).$ Then, $P_{M(\alpha)}$ is
an $\alpha$-generalized ordinal potential function satisfying
  \[ P_{M(\alpha)}(\xi)-
  P_{M(\alpha)}(\xi')\geq \big(\tfrac{\alpha}{2}\big)^{M(\alpha)} \text{ for
    all }(\xi,\xi')\in I^{\alpha}.\]
\end{lemma}
  \begin{proof}
    We must show that $P_{M(\alpha)}(\xi)-P_{M(\alpha)}(\nu_S,\xi_{-S})\geq
    \big(\tfrac{\alpha}{2}\big)^{M(\alpha)}$ for an arbitrary $\alpha$-improving move
    $(\xi,(\nu_S,\xi_{-S})) \in I^{\alpha}$.  Let $j\in S$ with
    $j~\in~\arg\max_{i \in S} \pi_i(\nu_S,\xi_{-S})$.
   
    We define

    $\Psi^+ := \{(i,f) \in -S \times F : \psi_{i,f}(\xi) \geq
    \pi_j(\nu_S,\xi_{-S})\}$ and $\Psi^- := \{(i,f) \in -S \times F :
    \psi_{i,f}(\xi) < \pi_j(\nu_S,\xi_{-S})\}$.

\begin{claim}

  \begin{enumerate}
\item 
 $\psi_{i,f}(\xi) \geq \psi_{i,f}(\nu_S,\xi_{-S})$ for all
  $(i,f)\in\Psi^+$
\item $\psi_{i,f}(\nu_S,\xi_{-S}) \leq \pi_j(\nu_S,\xi_{-S})$ for all
$(i,f)\in\Psi^-$.
\end{enumerate}
\end{claim}
\begin{proof}

  To prove the first claim, suppose there is $(k,g) \in \Psi^+$ such
  that $\psi_{k,g}(\xi) < \psi_{k,g}(\nu_S,\xi_{-S})$.  Because of the
  monotonicity of cost functions there exists $i\in S$ with $g\in
  F_i(g)$ implying
\begin{align*}
  \pi_j(\nu_S,\xi_{-S})\leq \psi_{k,g}(\xi)<\psi_{k,g}(\nu_S,\xi_{-S})\leq
  \pi_i(\nu_S,\xi_{-S}) \leq \pi_j(\nu_S,\xi_{-S}),
\end{align*}
which is a contradiction.  

For proving the second claim, suppose there is $(k,g) \in \Psi^-$ such
that $\psi_{k,g}(\nu_S,\xi_{-S}) >\pi_j(\nu_S,\xi_{-S})$.  Again,
monotonicity of cost functions implies that there is $i\in S$ with
$g\in F_i(g)$ giving rise to
\begin{align*}
  \pi_i(\nu_S,\xi_{-S})\geq \psi_{k,g}(\nu_S,\xi_{-S})> \pi_j(\nu_S,\xi_{-S})\geq
  \pi_i(\nu_S,\xi_{-S}),
\end{align*}
which is a contradiction. This proves the claim.
\qed\end{proof}

We observe that $N\times F=\Psi^+\cup\Psi^-\cup (F\times S)$.

Then,
\begin{align*}
  P_{M(\alpha)}(\xi)-P_{M(\alpha)}(\nu_S,\xi_{-S})
  &=\sum_{f\in \Psi^+\cup\Psi^-\cup (F\times S)} \psi_{i,f}(\xi)^{M (\alpha)}-\psi_{i,f}(\xi,y_{-S})^{M (\alpha)}\\
  &\geq \sum_{f\in \Psi^-\cup (F\times S)} \psi_{i,f}(\xi)^{M
    (\alpha)}-\psi_{i,f}(\nu_S,\xi_{-S})^{M (\alpha)}.
\end{align*}
The inequality follows from the first claim.  We further derive
\begin{align*}
\sum_{f\in \Psi^-\cup (F\times S)} \psi_{i,f}(\xi)^{M
     (\alpha)}-\psi_{i,f}(\nu_S,\xi_{-S})^{M (\alpha)}&\geq  \sum_{f\in
    (F\times S)} \psi_{i,f}(\xi)^{M (\alpha)}-\sum_{f\in \Psi^-\cup
    (F\times
    S)}\psi_{i,f}(\nu_S,\xi_{-S})^{M (\alpha)}\\
  &\geq (\pi_j(\nu_S,\xi_{-S})+\alpha)^{M
    (\alpha)}-(n\,m)\,\pi_j(\nu_S,\xi_{-S})^{M (\alpha)},
\end{align*}
where the first inequality follows from the non-negativity of $\psi$.
The second inequality follows from $\pi_j(\xi) \geq
\pi_j(\nu_S,\xi_{-S})+\alpha$ and the second claim.  To this end, we
obtain
\begin{align*}
  P_{M(\alpha)}(\xi)-P_{M(\alpha)}(\nu_S,\xi_{-S}) &\geq
  (\alpha/2)^{M(\alpha)}+(\pi_j(\nu_S,\xi_{-S})+\alpha/2)^{M
    (\alpha)}-(n\,m)\,\pi_j(\nu_S,\xi_{-S})^{M (\alpha)}\\
  &\geq (\alpha/2)^{M(\alpha)},
\end{align*}
where the last inequality follows from the choice of $M(\alpha)$.
\qed\end{proof}

\begin{proof}[Proof of Theorem~\ref{thm:approximate}]
  Fix $\alpha >0$. Then, since $\Delta$ is compact and $P_{M(\alpha)}$
  (as defined in Lemma~\ref{lem:approximate}) is bounded, there exists
  a strategy profile $z$ satisfying $P_{M(\alpha)}(z)\leq \inf_{\xi\in \Delta}
  P_{M(\alpha)}(\xi)-\epsilon$ with
  $0<\epsilon<\big(\frac{\alpha}{2}\big)^{M(\alpha)}$.  We claim that
  $z$ is an $\alpha$-approximate SNE.  Suppose not. Then there exists a profitable deviation
  $\nu_S\in \Delta_S$ with $
  P_{M(\alpha)}(z)-P_{M(\alpha)}(\nu_S,z_{-S})\geq(\alpha/2)^{M(\alpha)}>\epsilon$
  (by Lemma~\ref{lem:approximate}), which contradicts the
  approximation guarantee of $z$.
\qed\end{proof}

\end{document}